\documentclass[final,12pt]{article}
\usepackage{amsfonts,color,morefloats,pslatex,a4wide}
\usepackage{amssymb,amsthm,amsmath,latexsym,pslatex,cite}

\newtheorem{theorem}{Theorem}
\newtheorem{lemma}[theorem]{Lemma}

\newtheorem{definition}[theorem]{Definition}
\newtheorem{example}[theorem]{Example}

\usepackage[para]{threeparttable}

\newcommand{\gf}{{\mathrm{GF}}}

\newcommand{\wt}{{\mathtt{wt}}}

\newcommand{\CPRM}{{\mathtt{CPRM}}}

\newcommand{\F}{{\mathbb{F}}}
\newcommand{\Z}{\mathbb{{Z}}}

\newcommand{\m}{\mathbb{M}}

\newcommand{\C}{{\mathcal{C}}}

\newcommand{\ba}{{\mathbf{a}}}
\newcommand{\bb}{{\mathbf{b}}}
\newcommand{\bc}{{\mathbf{c}}}

\makeatletter

\newcommand{\Rmnum}[1]{\expandafter\@slowromancap\romannumeral #1@}
\makeatother

\begin{document}

\title{Two classes of constacyclic codes with a square-root-like lower bound
\thanks{Z. Sun's research was supported by The National Natural Science Foundation of China under Grant No. U21A20428. H. Chen's research was supported by The National Natural Science Foundation of China under Grant No. 62032009. C. Ding's research was supported by The Hong Kong Research Grants Council, Proj. No. $16301522$ and in part by the UAEU-AUA joint research grant G00004614. }}

\author{
Tingfang Chen 
\and Zhonghua Sun 
\and 
Conghui Xie 
\and 
Hao Chen                    
\and Cunsheng Ding 
\thanks{
Tingfang Chen is with the Department of Computer Science and Engineering, The Hong Kong University of Science and Technology, Clear Water Bay, Kowloon, Hong Kong, China. (e-mail: tchenba@connect.ust.hk). \newline 
Zhonghua Sun is with the School of Mathematics, Hefei University of Technology, Hefei, Anhui, 230601, China (e-mail:  sunzhonghuas@163.com). \newline 
Conghui Xie and Hao Chen are with the College of Information Science and Technology/Cyber Security, Jinan University, Guangzhou, Guangdong, 510632, China (e-mail: conghui@stu2021.jnu.edu.cn; haochen@jnu.edu.cn).\newline
Cunsheng Ding is with the Department of Computer Science and Engineering, The Hong Kong University of Science and Technology, Clear Water Bay, Kowloon, Hong Kong, China. (e-mail: cding@ust.hk).
}
}

\maketitle

\begin{abstract}
Constacyclic codes over finite fields are an important class of linear codes as they contain distance-optimal codes and linear codes with best known parameters. They are interesting in theory and practice, as they have the constacyclic structure. In this paper, an infinite class of $q$-ary negacyclic codes of length $(q^m-1)/2$ and an infinite class of $q$-ary constacyclic codes of length $(q^m-1)/(q-1)$ are constructed and analyzed. As a by-product, two infinite classes of ternary negacyclic self-dual codes with a square-root-like lower bound on their minimum distances are presented.

\vspace*{.3cm}

\end{abstract}

\newpage 
\section{Introduction}

Let $q$ be a prime power, $\gf(q)$ be the finite field with $q$ elements and $\gf(q)^*=\gf(q)\backslash \{0\}$. A $q$-ary $[n, k, d]$ linear code $\C$ is $k$-dimensional linear subspace of $\gf(q)^n$ with minimum distance $d$. Let $\lambda \in \gf(q)^*$. The linear code $\C$ is called a {\it $\lambda$-constacyclic} code if $(\lambda c_1,c_0,\ldots, c_{n-1})\in \C$ for any $(c_0,c_1,\ldots, c_{n-1})\in \C$. Define
\begin{align*}
\varphi:~\gf(q)^n&\rightarrow \gf(q)[x]/(x^n-\lambda)\notag	\\
(c_0,c_1,\ldots, c_{n-1})&\mapsto c_0+c_1x+\cdots+ c_{n-1}x^{n-1}.
\end{align*}
Then $\C$ is a $q$-ary $\lambda$-constacyclic code of length $n$ if and only if $\varphi(\C)$ is an ideal of the ring $\gf(q)[x]/(x^n-\lambda)$. Due to this, we will identify $\C$ with $\varphi(\C)$ for any $\lambda$-constacyclic code $\C$. It is known that $\gf(q)[x]/(x^n-\lambda)$ is a {\it principal ideal ring}. Let $\C=(g(x))$ be a $q$-ary $\lambda$-constacyclic code of length $n$, where $g(x)$ is monic and has the smallest degree among non-zero codewords of $\C$. Then $g(x)$ divides $x^n-\lambda$ and is called the {\it generator polynomial} of $\C$. The $h(x)=(x^n-\lambda)/g(x)$ is called the {\it check polynomial} of $\C$.

When $\lambda=1$, $\lambda$-constacyclic codes are classical cyclic codes. When $q$ is odd and $\lambda=-1$, $\lambda$-constacyclic codes are negacyclic codes. Constacyclic codes contain cyclic codes as a subclass. For certain lengths, constacyclic codes have better parameters than cyclic codes \cite{Geor82,SD22c}. Several infinite classes of distance-optimal constacyclic codes were constructed in \cite{DDR11,SHZ23,ZKZL}. Several infinite classes of MDS codes were constructed in \cite{KS90,WDLZ}. Several infinite classes of constacyclic codes with best known parameters were constructed in \cite{CW}. Several infinite classes of binary distance-optimal linear codes were constructed from quaternary constacyclic codes \cite{SHZ23}. Several infinite classes of negacyclic self-dual codes with a square-root-like lower bound were constructed in \cite{SD22c,XCDS}. Recently, Sun et al. constructed two classes of constacyclic codes with very good parameters from Dilix cyclic codes and punctured generalised Reed-Muller codes \cite{SDW}.

Binary Reed-Muller codes were introduced by Reed and Muller \cite{Muller54,Reed54}. If the Reed-Muller codes are punctured in a special coordinate position, the punctured Reed-Muller codes are cyclic. The binary Reed-Muller codes and their punctured codes were later generalised to the $q$-ary case \cite{DGM70,KLP,Lachaud,Sorensen,WW}. Recently, with the punctured binary Reed-Muller codes, Wu et al. constructed binary cyclic codes with best parameters known \cite{WSD}. Sun et al. constructed a class of constacyclic codes that are scalar-equivalent to the projective Reed-Muller codes \cite{SDW}. Inspired by these two works, we use projective Reed-Muller codes to construct constacyclic codes.

The objectives of this paper are the following:
\begin{enumerate}
\item Construct and analyze two classes of constacyclic codes that contain distance-optimal codes.	
\item Study the minimum distances of the duals of these constacyclic codes.
\end{enumerate}
The main contributions of this paper are the following:
\begin{enumerate}
\item Two infinite classes of constacyclic codes $\C$ such that the minimum distances $d(\C)$ and $d(\C^\perp)$ both have a very good lower bound are constructed.
\item Two infinite classes of ternary negacyclic self-dual codes with very good lower bounds on their minimum distances are obtained. 	
\end{enumerate}

The rest of this paper is organized as follows. In Section \ref{sec2}, we present some auxiliary results. In Section \ref{sec3}, we construct and study the first class of constacyclic codes. An infinite class of ternary negacyclic self-dual codes with a square-root-like lower bound is presented. In Section \ref{sec4}, we construct and study the second class of constacyclic codes. In Section \ref{sec5}, we study the subcodes of the second class of constacyclic codes. Another infinite class of ternary negacyclic self-dual codes with a square-root-like lower bound is constructed. In Section \ref{sec6}, we conclude this paper and make some concluding remarks.

\section{Preliminaries}\label{sec2}

Let $q>2$ be a prime power, $r\mid (q-1)$ and $m\geq 2$ be an integer. Let $n=(q^m-1)/r$ and $N=rn=q^m-1$. Let $\Z_{N}=\{0,1,\cdots, N-1\}$ be the ring of integers modulo $N$. For any integer $i$, let $i \bmod N$ denote the unique $h$ with $0\leq h\leq N-1$ such that $h\equiv i~({\rm mod}~N)$. For any $i\in \Z_N$, the {\it $q$-cyclotomic coset} of $i$ modulo $N$ is defined by
$$
C_i^{(q,N)}=\left\{ i q^{j} \bmod N: ~0\leq j\leq \ell_i-1 \right \},	
$$
where $\ell_i$ is the smallest positive integer such that $iq^{\ell_i}\equiv i~({\rm mod}~N)$. The smallest integer in $C_i^{(q,N)}$ is called the {\it coset leader} of $C_i^{(q,N)}$. Let $\Gamma_{(q, r, N)}$ be the set of all coset leaders,
$$\Gamma_{(q, r, N)}^{(1)}=\left \{ i\equiv 1~({\rm mod}~r): i\in \Gamma_{(q,r,N)} \right \}$$
 and $\Omega_{(r,n)}^{(1)}=\{1+ri:~0\leq i\leq n-1 \}$. Then
\begin{equation}\label{EQ::4}
\Omega_{(r, n)}^{(1)}=\bigcup_{i\in \Gamma_{(q, r, N)}^{(1)}}C_i^{(q, N)}.	
\end{equation}

Let $\beta$ be a primitive $N$-th root of unity and $\lambda=\beta^n$. Then $\lambda$ is an element of $\gf(q)$ with order $r$. For each $i\in \Z_N$, let $\m_{\beta^i}(x)$ denote the minimal polynomial of $\beta^i$ over $\gf(q)$. Then
\begin{equation}\label{EQ::6}
\m_{\beta^i}(x)=\prod_{j\in C_i^{(q,N)}} (x-\beta^j),	
\end{equation}
which is irreducible over $\gf(q)$. By Equations (\ref{EQ::4}) and (\ref{EQ::6}), we get that
\begin{equation}\label{EQ::7}
x^n-\lambda=\prod_{i\in \Gamma_{(q, r, N)}^{(1)}} \m_{\beta^i}(x),
\end{equation}
which is the factorization of $x^n-\lambda$ into irreducible factors over $\gf(q)$. For a $q$-ary $\lambda$-constacyclic code $\C$ of length $n$, the set
\begin{equation*}\label{EQ::8}
Z(\C)=\left \{i\in \Omega_{(r,n)}^{(1)}:~g(\beta^i)=0 \right \}	
\end{equation*}
is called the {\it defining set} of $\C$ with respect to the primitive $N$-th root of unity $\beta$. By Equation (\ref{EQ::7}), we get that $Z(\C)$ is the union of some $q$-cyclotomic cosets modulo $N$.

For any $q$-ary linear code $\C$, we use $\dim(\C)$ and $d(\C)$ to denote its dimension and minimum Hamming distance, respectively.
\begin{lemma} \label{lem::1}
\cite[Lemma 4]{KS90} [Constacyclic BCH bound]
 If there are integer $b$ with $b\equiv 1~({\rm mod}~r)$, integer $a$ with $\gcd(a, N)=r$, integer $h$ and integer $\delta$ with $2\leq \delta \leq n$ such that
$$
		\left \{(b+a i) \bmod N:~h\leq i \leq h+\delta-2 \right \}\subseteq Z(\C),
$$
	then $d(\C)\geq \delta$.
\end{lemma}

For a $q$-ary linear code $\C$ of length $n$, the {\it dual code} of $\C$ is defined by
$$
\C^\perp=\left \{\ba \in \gf(q)^n: \ba \bb^t=0~{\rm for~ all~}\bb \in \C \right \},	
$$
where $\bb^t$ denotes the transposition of the vector $\bb$. Let $f(x)=f_0+f_1x+\cdots+f_s x^s \in \gf(q)[x]$, where $f_0 f_s\neq 0$ and $s$ is a positive integer. The {\it reciprocal polynomial} of $f(x)$, denoted by $\widehat{f}(x)$, is defined by $\widehat{f}(x)=f_0^{-1}x^s f(x^{-1})$. According to \cite[Lemma 1]{KS90}, we have the following result.

\begin{lemma}\label{lem::2}
Let $\C$ be a $q$-ary $\lambda$-constacyclic code with length $n$ and check polynomial $h(x)$. Then $\C^\perp$ is a $q$-ary $\lambda^{-1}$-constacyclic code with length $n$ and generator polynomial $\widehat{h}(x)$.
\end{lemma}

Let $\C$ be a $q$-ary $\lambda$-constacyclic code with length $n$ and generator polynomial $g(x)$. We define the following two related codes of the code $\C$.
\begin{enumerate}
\item The {\it complement code}, denoted by $\C^c$, of $\C$ is the $\lambda$-constacyclic code with length $n$ and generator polynomial $h(x)$, where $h(x)=(x^n-\lambda)/g(x)$.
\item The {\it reverse code}, denoted by $\C^{-}$, of $\C$ is the $\lambda^{-1}$-constacyclic code with length $n$ and generator polynomial $\widehat{g}(x)$.
\end{enumerate}
It is easily verified that the following hold.
\begin{enumerate}
\item The codes $\C$ and $\C^{-}$ have the same parameters.	
\item $\C^\perp=(\C^c)^-$.
\item The codes $\C^\perp$ and $\C^c$ have the same parameters.
\end{enumerate}

 According to the constacyclic BCH bound, we have the following result.

\begin{lemma} \label{lem::3}
	If there are integer $b$ with $b\equiv 1~({\rm mod}~r)$, integer $a$ with $\gcd(a, N)=r$, integer $h$ and integer $\delta$ with $2\leq \delta \leq n$ such that
	\begin{equation*}\label{EQ::11}
		\left \{(b+a i) \bmod N:~h\leq i \leq h+\delta-2 \right \}\subseteq (\Omega_{(r, n)}^{(1)} \backslash Z(\C)),
	\end{equation*}
	then $d(\C^\perp)\geq \delta$.
\end{lemma}

\begin{proof}
The desired result follows from the fact that $\C$ and $\C^c$ have the same parameters.	
\end{proof}

For any $0\leq i\leq N=q^m-1$, the $q$-adic expansion of $i$ is defined by
$$
i_0+i_1q +\cdots+i_{m-1}q^{m-1},	
$$
where $i_0, i_1, \cdots, i_{m-1}\in \{0,1,\cdots, q-1\}$. There are two types of weight definitions as follows.
\begin{enumerate}
\item The {\it $q$-weight} of $i$, denoted by $\wt_q(i)$, is defined by $i_0+i_1+\cdots+i_{m-1}$.
\item  The {\it Hamming weight} of $i$, denoted by $\wt(i)$, is defined by the Hamming weight of the vector $(i_0,i_1, \ldots, i_{m-1})$. 	
\end{enumerate}
It is easily verified that $\wt_q(i)$ and $\wt(i)$ are constants on each $q$-cyclotomic coset $C_b^{(q,N)}$, respectively.

When $r=q-1$ and $n=(q^m-1)/(q-1)$, for any $0\leq \ell\leq m-2$, the constacyclic projective Reed-Muller code, denoted by $\CPRM(q, m, m-2-\ell)$, is the $q$-ary $\lambda$-constacyclic code with length $n=(q^m-1)/(q-1)$ and defining set
\begin{equation}\label{definingset1}
D_{(q, m,\ell)}=\bigcup_{i\in \Gamma_{(q, r, N)}^{(1)} \atop \wt_q(i)<1+(q-1)(\ell+1) }C_i^{(q, N)}
\end{equation}
with respect to the primitive $N$-th root of unity $\beta$. In the following we will use constacyclic projective Reed-Muller codes to construct new constacyclic codes with very good parameters. We need the following lemma.

\begin{lemma}\cite[Corollary 48]{SDW}\label{SDW}
Let $q>2$ be a prime power, $m\geq 2$ be an integer and $0\leq \ell \leq m-2$. Then $d(\CPRM(q, m, m-2-\ell))=3  \cdot q^{\ell}$.
\end{lemma}

\section{The first class of constacyclic codes} \label{sec3}

We follow the notation in the previous sections, but fix $r=2$, let $q$ be odd, $n=(q^m-1)/2$ and $N=rn=q^m-1$ throughout this section. In the following we present a split of
$$\Omega_{(r, n)}^{(1)}:=\{1+2i:0\leq i\leq n-1\}$$
based on the Hamming weight. Subsequently, negacyclic codes with very good parameters are constructed.

\begin{lemma} \label{lem:1}
For $i \in \{0,1\}$, define
\begin{eqnarray*}\label{dsun:2}
T_{(i, n)}=\left\{ h:~\wt(h) \equiv i~({\rm mod}~2), ~h\in \Omega_{(r, n)}^{(1)} \right \}.
\end{eqnarray*}
Then $T_{(i, n)}$ is the union of some $q$-cyclotomic cosets, and $\Omega_{(r, n)}^{(1)}=T_{(0,n)} \cup T_{(1, n)}$.
\end{lemma}

\begin{proof}
It is easily verified that $\wt(h)$	is a constant on each $q$-cyclotomic coset $C_h^{(q, N)}$. Then, for each $i\in \{0,1\}$, $T_{(i, n)}$ is the union of some $q$-cyclotomic cosets. The second desired result is obvious.
\end{proof}

For $i\in \{0,1\}$, let $\C_{(i, n)}$ be the $q$-ary negacyclic code of length $n$ with defining set $T_{(i, n)}$, and let $\C_{(i, n)}^{\perp}$ be the dual of $\C_{(i, n)}$. Since $\Omega_{(r, n)}^{(1)}=T_{(0,n)} \cup T_{(1, n)}$, we can deduce that $\C_{(i, n)}^{\perp}$ and $\C_{(i\oplus 1, n)}$ have the same parameters, where $\oplus$ denotes the modulo-2 addition.
According to the Magma experimental data in Table \ref{DS:1}, the codes $\C_{(1, n)}$ and $\C_{(1,n)}^{\perp}$ have very good parameters. Although $\C_{(1,n)}^\perp$ or $\C_{(1, n)}$ does not achieve the best known minimum distance compared with linear codes of the same length and dimension sometimes, it has the negacyclic structure (this is an advantage).  Hence, the two codes $\C_{(1,n)}$ and $\C_{(1,n)}^\perp$ are interesting. This is the motivation for constructing and studying these negacyclic codes.

\begin{table}
\begin{center}
\caption{Examples of codes $\C_{(1,n)}$ and $\C_{(1,n)}^\perp$}
\label{DS:1}
\begin{tabular}{|c|c|c|c|c|c|}
\hline   $q$ & $m$ & $\C_{(1,n)}$ & Optimality & $\C_{(1,n)}^\perp$ & Optimality\\
\hline  $3$  & $2$ & $[4,2,3]$  & optimal linear code & $[4,2,3]$ & optimal linear code   \\
\hline  $3$  & $3$ & $[13,6,6]$  & optimal linear code & $[13,7,5]$ &  optimal linear code  \\
\hline  $3$  & $4$ & $[40,20,9]$  & $d_{\rm best}=12$ & $[40,20,9]$ & $d_{\rm best}=12$   \\
\hline  $5$  & $2$ & $[12,8,4]$  & optimal linear code & $[12,4,6]$ &  $d_{\rm optimal}=8$ \\
\hline  $5$  & $3$ & $[62,24,22]$  & $d_{\rm best}=23$ & $[62,38,12]$ & $d_{\rm best}=14$  \\
\hline  $7$  & $2$ & $[24,18,5]$  & best linear code known & $[24,6,14]$ & $d_{\rm best}=16$  \\
\hline  $9$  & $2$ & $[40,32,6]$  & best linear code known & $[40,8,20]$ & $d_{\rm best}=27$   \\
\hline
\end{tabular}
\end{center}
\end{table}

Next, we investigate the parameters of these two classes of negacyclic codes.
It follows from $\Omega_{(r, n)}^{(1)}=T_{(0,n)} \cup T_{(1, n)}$ and $T_{(0,n)} \cap  T_{(1, n)}=\emptyset$ that
\begin{equation}\label{ds::3}
	\dim(\C_{(i, n)})=|T_{(i\oplus 1,n)}|
\end{equation}
 for each $i\in \{0,1\}$. To determine $\dim(\C_{(i, n)})$, it suffices to determine $|T_{(i, n)}|$. We will need the following lemma.

\begin{lemma}\cite{SDW} \label{lem:2}
Let $t$ be a positive integer. Then the number of solutions $(x_1,x_2,\ldots, x_t )$ with $1\leq x_i\leq q-1$ to the congruence $x_1+x_2+\cdots+x_t\equiv 1~({\rm mod}~2)$ is equal to $(q-1)^t/2$.	
\end{lemma}

Let $0\leq i\leq q^m-1$ and $i=i_0+i_1q+\cdots+i_{m-1}q^{m-1}$, where $0\leq i_j\leq q-1$. Define $${\rm zs}_\epsilon(i)=|\{h: i_h=\epsilon,~0\leq h\leq m-1 \}|,$$
where $0\leq \epsilon\leq q-1$. By definition, $\wt(i)=m-{\rm zs}_0(i)$ and $\wt(q^m-1-i)=m-{\rm zs}_{q-1}(i)$. Furthermore, we have the following results.

\begin{lemma}\label{lem:3}
Let notation be the same as before. The following hold.
\begin{enumerate}
\item $|T_{(1,n)}|=[q^m-(2-q)^m]/4$ and $|T_{(0,n)}|=[q^m+(2-q)^m-2]/4$.
\item Let $q=3$. Then $T_{(1, n)}=\{q^m-1-i: i\in T_{(0,n)} \}$ if $m$ is even, and for any $i\in T_{(h, n)}$, $q^m-1-i\in T_{(h, n)}$ if $m$ is odd.
\end{enumerate}
\end{lemma}

\begin{proof}
1. By definition,
\begin{align*}
|T_{(1,n)}|&=|\{ 1\leq i\leq q^m-1: {\rm zs}_0(i)\equiv m+1~({\rm mod}~2),~i\equiv 1~({\rm mod}~2)\} |\\
&=\begin{cases}
	\sum_{0\leq a\leq m~{\rm is~even}} \binom{m}{a} \frac{(q-1)^{m-a}}2 & ~{\rm if~}m~{\rm is~odd},\\
	\sum_{1\leq a\leq m~{\rm is~odd}} \binom{m}{a} \frac{(q-1)^{m-a}}{2} & ~{\rm if~}m~{\rm is~even}
\end{cases}\\
&=[q^m-(2-q)^m]/4.
\end{align*}

2. Let $q=3$ and $i\in \Omega_{(2,n)}^{(1)}$. Then ${\rm zs}_1(i)\equiv i\equiv 1~({\rm mod}~2)$ and  ${\rm zs}_0(i)+{\rm zs}_2(i)\equiv m-1~({\rm mod}~2)$. By defintion,
\begin{align*}
\wt(3^m-1-i)=m-{\rm zs}_{2}(i)\equiv 1+{\rm zs}_0(i) ~({\rm mod}~2)\equiv \begin{cases}
 	1+\wt(i) ~({\rm mod}~2) ~&{\rm if}~m~{\rm is~even},\\
 	\wt(i)~({\rm mod}~2) ~&{\rm if}~m~{\rm is~odd}.\\
 \end{cases}	
\end{align*}
The desire results then follow. This completes the proof. 	
\end{proof}

The result documented in the following lemma is well known.

\begin{lemma}
Let $h$ be a positive integer and let $m/\gcd(h, m)$ be odd. Then
$$\gcd(q^h+1, q^m-1)=2.$$	
\end{lemma}

When $m\geq 3$ is odd, the parameters of the code $\C_{(1,n)}$ and its dual are documented in the following theorem.

\begin{theorem}\label{thm::6}
Let $m\geq 3$ be odd. Then $\C_{(1,n)}$ has parameters
$$\left[(q^m-1)/2, \, [q^m+(2-q)^m-2]/4, \, d \geq q^{(m-1)/2}+1+\epsilon  \right],$$
where $\epsilon =0$ if $q=3$, and $\epsilon=q-2$ if $q\geq 5$. The dual code $\C_{(1,n)}^\perp$ has parameters	
$$\left[(q^m-1)/2, \, [q^m-(2-q)^m]/4, \, d^\perp \geq (q-1)q^{(m-3)/2}+1 \right].$$
\end{theorem}

\begin{proof}
The dimensions of $\C_{(1, n)}$ and $\C_{(1,n)}^\perp$ follow from Equation (\ref{ds::3}) and Lemma \ref{lem:3}. We now lower bound the minimum distances of $\C_{(1, n)}$ and $\C_{(1,n)}^\perp$.

Let $h=(m-1)/2$. Then $\gcd(q^h+1, q^m-1)=2$. Define
$$H=\{(q-2)q^{2h}+(q^h+1)i:- \epsilon \leq i\leq q^h-1\}.$$
Below we will prove that $H\subseteq T_{(1, n)}$ in the following two cases.
\begin{itemize}
\item  Let $0\leq i\leq q^h-1$. 	Clearly,
$$\wt((q-2)q^{2h}+(q^h+1)i)=2\wt(i)+1\equiv 1~({\rm mod}~2).$$
\item Let $q\geq 5$ and $-(q-2)\leq i\leq -1$. In this case,
$$(q-2)q^{2h}+(q^h+1)i=(q-3)q^{2h}+(q-1)\sum_{j=1, j\neq h}^{2h-1}q^j+(q-1+i)q^h+q+i.$$
Then $\wt((q-2)q^{2h}+(q^h+1)i)=2h+1\equiv 1~({\rm mod}~2)$.
\end{itemize}
Therefore, $H\subseteq T_{(1, n)}$. By the constacyclic BCH bound, $d(\C_{(1, n)})\geq q^h+1+\epsilon$.

Define $H^\perp=\{2q^{2h-1}+q^{h-1}+(q^h+1)i: 0\leq i\leq (q-1)q^{h-1}-1 \}$. Suppose $i=i_1 q^{h-1}+i_0$, where $0\leq i_0\leq q^{h-1}-1$ and $0\leq i_1\leq q-2$. Then
\begin{align*}
	2q^{2h-1}+q^{h-1}+(q^h+1)i=(2+i_1)q^{2h-1}+i_0q^h+(i_1+1)q^{h-1}+i_0.
\end{align*}
It follows that
$$\wt(2q^{2h-1}+q^{h-1}+(q^h+1)i)=2+2\wt(i_0)\equiv 0~({\rm mod}~2).$$
Therefore, $H^\perp \subseteq T_{(0, n)}$. By the constacyclic BCH bound, $d(\C_{(0, n)})\geq (q-1)q^h+1$. Notice that $d(\C_{(1, n)}^\perp )=d(\C_{(0, n)})$. This completes the proof.	
\end{proof}

When $q\geq 5$ and $m$ is not a power of $2$, the parameters of the code $\C_{(1,n)}$ and its dual are documented in the following theorem.

\begin{theorem}\label{thm::7}
Let $q\geq 5$ be an odd prime power and $m=2^e \ell$, where $\ell\geq 3$ is odd  and $e$ is a positive integer. Then $\C_{(1,n)}$ has parameters
$$\left[(q^m-1)/2, \, [q^m+(2-q)^m-2]/4, \, d \geq q^{(m-2^e)/2}+q  \right],$$
and the dual code $\C_{(1,n)}^\perp$ has parameters	
$$\left[(q^m-1)/2, \, [q^m-(2-q)^m]/4, \, d^\perp \geq  q^{(m-2^e)/2}+q \right].$$
\end{theorem}

\begin{proof}
The dimensions of $\C_{(1, n)}$ and $\C_{(1,n)}^\perp$ follow from Equation (\ref{ds::3}) and Lemma \ref{lem:3}. We now lower bound the minimum distances of $\C_{(1, n)}$ and $\C_{(1,n)}^\perp$.

Let $h=(m+2^e)/2$. Then $\gcd(q^h+1,q^m-1)=2$. Define
$$H=\{[q^{h-2^e+1}+(q^h+1)i] \bmod N :-(q-2)\leq i\leq q^{h-2^{e}} \}.$$
Below we will prove that $H\subseteq T_{(1, n)}$ in the following three cases.
\begin{itemize}
\item Let $i=q^{h-2^e}$. In this case, $[q^{h-2^e+1}+(q^h+1)i] \bmod N =q^{h-2^e+1}+q^{h-2^e}+1$. Then
$$\wt([q^{h-2^e+1}+(q^h+1)i] \bmod N)=3\equiv 1~({\rm mod}~2).$$
\item Let $0\leq i\leq q^{h-2^e}-1$. In this case, $[q^{h-2^e+1}+(q^h+1)i] \bmod N =iq^h+q^{h-2^e+1}+i$. Then  	
$$\wt([q^{h-2^e+1}+(q^h+1)i] \bmod N)=2\wt(i)+1\equiv 1~({\rm mod}~2).$$
\item Let $-(q-2)\leq i\leq -1$. In this case,
$$
	[q^{h-2^e+1}+(q^h+1)i] \bmod N =q^m-1+q^{h-2^e+1}+(q^h+1)i.
$$
It follows that
$$\wt([q^{h-2^e+1}+(q^h+1)i] \bmod N)=m-2^e+1\equiv 1~({\rm mod}~2).$$
\end{itemize}
Therefore, $H\subseteq T_{(1, n)}$. By the constacyclic BCH bound, $d(\C_{(1, n)})\geq q^{(m-2^e)/2}+q$.

Define
$$H^\perp=\{[q^{h-2^e+1}+(q-1)q^{h-2^e}+(q^h+1)i] \bmod N :-(q-2)\leq i\leq q^{h-2^{e}} \}.$$
Below we will prove that $H^\perp \subseteq T_{(0, n)}$ in the following three cases.
\begin{itemize}
\item Let $i=q^{h-2^e}$. In this case, $$[q^{h-2^e+1}+(q-1)q^{h-2^e}+(q^h+1)i] \bmod N =2q^{h-2^e+1}+1.$$
Then
$$\wt([q^{h-2^e+1}+(q-1)q^{h-2^e}+(q^h+1)i] \bmod N )=2\equiv 0~({\rm mod}~2).$$
\item Let $0\leq i\leq q^{h-2^e}-1$. In this case,
\begin{align*}
[q^{h-2^e+1}+(q-1)q^{h-2^e}+(q^h+1)i] \bmod N =iq^h+q^{h-2^e+1}+(q-1)q^{h-2^e}+i.	
\end{align*}
Then  	
$$\wt([q^{h-2^e+1}+(q-1)q^{h-2^e}+(q^h+1)i] \bmod N)=2\wt(i)+2\equiv 0~({\rm mod}~2).$$
\item Let $-(q-2)\leq i\leq -1$. In this case,
\begin{align*}
	[q^{h-2^e+1}+(q-1)q^{h-2^e}+(q^h+1)i] \bmod =~q^m-1+q^{h-2^e+1}+(q-1)q^{h-2^e}+(q^h+1)i.
\end{align*}
It follows that
$$\wt([q^{h-2^e+1}+(q-1)q^{h-2^e}+(q^h+1)i] \bmod N)=m-2^e+2\equiv 0~({\rm mod}~2).$$
\end{itemize}
Therefore, $H^\perp \subseteq T_{(0, n)}$. By the constacyclic BCH bound, $d(\C_{(0, n)})\geq q^{(m-2^e)/2}+q$. Notice that $d(\C_{(1, n)}^\perp )=d(\C_{(0, n)})$. This completes the proof.
\end{proof}

When $q=3$ and $m$ is even, the parameters of the code $\C_{(1,n)}$ and its dual are documented in the following theorem.

\begin{theorem}\label{thm:8}
Let $q=3$ and $m\geq 4$ be even. Then $\C_{(1,n)}$ is a ternary negacyclic self-dual code and has parameters $[(3^m-1)/2, \,  (3^m-1)/4, \, d\geq d_m]$, 	where
$$d_m=\begin{cases}
3^{(m-2)/2}+3~&{\rm if}~m\equiv 2~({\rm mod}~4),\\
[3^{(m-2)/2}+15]/2~ &{\rm if}~	m\equiv 0~({\rm mod}~4).
\end{cases}
$$
\end{theorem}

\begin{proof}
By Result 2 of Lemma \ref{lem:3}, $T_{(1, n)}=\{q^m-1-i: i\in T_{(0,n)} \}$. It follows that $\C_{(1,n)}$ is a ternary self-dual code. We now lower bound the minimum distance of $\C_{(1, n)}$.

{\it Case 1}: Let $m\equiv 2\pmod{4}$. Let $h=(m+2)/2$. Then $\gcd( 3^h+1, 3^m-1)=2$. Define
$$H=\{[3^{h-1}+(3^h+1)i] \bmod N: -1\leq i\leq 3^{h-2}   \}.$$
Below we will prove that $H\subseteq T_{(1, n)}$ in the following three subcases.
\begin{itemize}
\item Let $i=3^{h-2}$. In this case, $[3^{h-1}+(3^h+1)i] \bmod N=3^{h-1}+3^{h-2}+1$. Then $$\wt([3^{h-1}+(3^h+1)i] \bmod N)=3\equiv 1~({\rm mod}~2).$$	
\item Let $0\leq i\leq 3^{h-2}-1$. In this case, $[3^{h-1}+(3^h+1)i] \bmod N=i\cdot 3^{h}+3^{h-1}+i$. Then
$$\wt([3^{h-1}+(3^h+1)i] \bmod N)=2\wt(i)+1\equiv 1~({\rm mod}~2).$$	
\item Let $i=-1$. In this case, $[3^{h-1}+(3^h+1)i] \bmod N=3^m-3^{h}+3^{h-1}-2$. Then
 $$\wt([3^{h-1}+(3^h+1)i] \bmod N)=m-1\equiv 1~({\rm mod}~2).$$
\end{itemize}
Therefore, $H\subseteq T_{(1, n)}$. By the constacyclic BCH bound, $d(\C_{(1, n)})\geq 3^{h-2}+3$.

{\it Case 2}: Let $m\equiv 0\pmod{4}$. Let $h=(m-2)/2$. Then $h$ is odd. Let $b=3^{m-1}-4(3^{h}-1)$ and $v=n+3^{h}-1$, then $b$ is odd and $v\equiv 2\pmod{4}$. Notice that $\gcd(v, n)=(3^h-1, n)$ and
$$\gcd(3^h-1, 3^m-1)=3^{\gcd(h,m)}-1=2,$$
then $\gcd(v/2, n)=1$. Define
$$H=\{(b+vi) \bmod N: 1\leq i\leq  (3^h+13)/2\}.$$
If $m=4$, it is easily checked that $H\subseteq T_{(1, n)}$. If $m\geq 8$ and $m\equiv 0\pmod{4}$, we will prove that $H\subseteq T_{(1, n)}$ in the following two cases.
\begin{itemize}
\item Let $i$ be even. In this case,
       $$(b+vi) \bmod N= 3^{m-1}+(3^h-1)(i-4).$$
 It follows that the following hold.
\begin{itemize}
\item If $i=2$, $\wt((b+vi) \bmod N)=(m+2)/2 \equiv 1~({\rm mod}~2)$.
\item If $i=4$, $\wt((b+vi) \bmod N)=1$.
\item If $6\leq i\leq (3^h+13)/2$, $1\leq i-5<3^h-1$ and
$$(b+vi) \bmod N= 3^{m-1}+(i-5)3^h+3^h-1-(i-5).$$
Suppose $i-5=i_0+i_1 3+\cdots+i_{h-1} 3^{h-1}$, where $i_j\in \{0,1,2\}$. Then
$$\wt((b+vi) \bmod N )=1+(h-A)+(h-B),$$ where
 \begin{align*}
 A&= |\{0\leq \ell \leq h-1: i_\ell=0 \}|,\\
 B&=|\{0\leq \ell \leq h-1: i_\ell=2 \}|.	
 \end{align*}
 Clearly,
 \begin{align*}
 i-5&\equiv i_0+i_1+\cdots+i_{h-1} ~({\rm mod}~2)\\
 	&\equiv 0\cdot A+(h-A-B)+2\cdot B ~({\rm mod}~2)\\
 	&\equiv 1-A-B ~({\rm mod}~2).
 \end{align*}
 Therefore, $A+B\equiv 0~({\rm mod}~2)$. Consequently,
 $$\wt((b+vi) \bmod N )\equiv 1~({\rm mod}~2).$$
\end{itemize}

\item Let $i$ be odd. In this case,
 $$(b+vi) \bmod N= n+3^{m-1}+(3^h-1)(i-4).$$
 It follows that the following hold.
 \begin{itemize}
\item If $1\leq i\leq (3^h+7)/2$,
$$(b+vi) \bmod N= 2\cdot 3^{m-1}+3^{m-2}+\left[ \frac{3^h-1}2+i-4  \right]3^h+\frac{3^h-1}2-(i-4).$$
Let $(3^h-1)/2+i-4=j$. Then $(3^h-7)/2\leq j\leq 3^h-1$ is even. Consequently,
  $$(b+vi) \bmod N= 2\cdot 3^{m-1}+3^{m-2}+ j 3^h+(3^h-1-j).$$
Suppose $j=j_0+j_1 3+\cdots+j_{h-1} 3^{h-1}$, where $j_l\in \{0,1,2\}$. Then
$$\wt((b+vi) \bmod N )=2+(h-A')+(h-B'),$$ where
 \begin{align*}
 A'&= |\{0\leq \ell \leq h-1: j_\ell=2 \}|,\\
 B'&=|\{0\leq \ell \leq h-1: j_\ell=0 \}|.	
 \end{align*}
 Clearly, $j\equiv 1-A'-B' ~({\rm mod}~2)$. Therefore, $A'+B' \equiv 1 ~({\rm mod}~2)$. Consequently,
 $$\wt((b+vi) \bmod N )\equiv 1~({\rm mod}~2).$$
 \item If $i=(3^h+11)/2$, $(b+vi) \bmod N= 2\cdot 3^{m-1}+2\cdot 3^{m-2}+3^h-2$. Then
 $$\wt((b+vi) \bmod N )= h+2 \equiv 1 ~({\rm mod}~2).$$
\end{itemize}
Therefore, $H\subseteq T_{(1, n)}$. By the constacyclic BCH bound, $d(\C_{(1, n)})\geq (3^{h}+15)/2$.
\end{itemize}
This completes the proof. 	
\end{proof}

The ternary negacyclic self-dual codes $\C_{(1, n)}$ have very good parameters. The example codes below justify this claim.

\begin{example}
Let $m=2$ and $\beta$ be a primitive element of $\F_{3^2}$ with $\beta^2+2\beta+2=0$. Then $\C_{(1,n)}$ is a ternary $[4,2,3]$ self-dual code and has weight enumerator $1+8z^3$. This code is MDS and optimal.
\end{example}

\begin{example}\label{exam-4m}
Let $m=4$ and $\beta$ be a primitive element of $\F_{3^4}$ with $\beta^4+2\beta^3+2=0$. Then $\C_{(1, 40)}$ is a ternary $[40,20,9]$ self-dual code and has weight enumerator $1+1040z^9+18720z^{12}+1100736z^{15}+ 25761840z^{18}+236377440z^{21}+908079120z^{24}+1388750720z^{27}+783679104z^{30}+137535840z^{33}+5468320z^{36}+11520z^{39}$.  The best ternary self-dual code known of length $40$ and dimension $20$ has minimum distance $12$ \cite{self-dual}, but it is not known to be negacyclic. The advantage of this ternary $[40,20,9]$ self-dual code $\C_{(1, 40)}$ over
the best ternary self-dual $[40,20,12]$ is its negacyclic structure. Notice that  ternary cyclic self-dual codes do not exist \cite{jlx}.
\end{example}

\section{The second class of constacyclic codes}\label{sec4}

We follow the notation in the previous sections, but fix the following notation, unless it is stated otherwise:
\begin{enumerate}
\item $q>2$ is a prime power.
\item $m\geq 2$ is an integer.
\item $r=q-1$.
\item $n=(q^m-1)/r$.
\item $N=rn=q^m-1$.
\item $\lambda$ is a primitive element of $\gf(q)$.
\item $\beta$ is a primitive element of $\gf(q^m)$ such that $\beta^n=\lambda$.	
\end{enumerate}

\begin{definition}\label{Def::1}
For any $0\leq \ell \leq m-1$, define
\begin{equation}\label{definingset2}
T_{(q, m, \ell)}=\{ i\in \Z_N:~\wt_q(i)=1+(q-1)\ell \}	
\end{equation}	
and
$$
	g_{(q, m, \ell)}(x)=\prod_{i \in T_{(q, m, \ell)}}(x-\beta^i).
$$
\end{definition}

Since $\wt_q(i)$ is a constant on each $q$-cyclotomic coset $C_b^{(q,N)}$, we get that $T_{(q, m, \ell)}$ is the union of some $q$-cyclotomic cosets modulo $N$. Consequently, $g_{(q, m, \ell)}(x)\in \gf(q)[x]$. Furthermore,
 \begin{equation*}\label{EQ::15}
 x^n-\lambda=\prod_{\ell=0}^{m-1}g_{(q, m, \ell)}(x)	.
 \end{equation*}

Let $\C_{(q, m, \ell)}$ be a $q$-ary $\lambda$-constacyclic code with length $n$ and generator polynomial $g_{(q, m, \ell)}(x)$. This class of constacyclic codes are very interesting due to the following facts:
\begin{enumerate}
\item The code $\C_{(q, m , 0)}$ is the $q$-ary Hamming code and $\C_{(q, m, 0)}^\perp$ is the $q$-ary Simple code.	Both $\C_{(q, m , 0)}$ and $\C_{(q, m, 0)}^\perp$ are optimal.
\item When $q\in \{3,4,5,7\}$ and $m\in \{2,3\}$, both $\C_{(q, m,\ell)}$ and $\C_{(q, m, \ell)}^\perp$ are optimal or have best known parameters, except for seven examples (see Table \ref{Tab-code1}).
\end{enumerate}

For any $0\leq i\leq N=q^m-1$, it is clear that $\wt_q(N-i)=(q-1)m-\wt_q(i)$. It then follows that $\wt_q(N-i)=q-2+(q-1)(m-1-\ell)$ for any $i\in T_{(q, m, \ell)}$. Then we have the following results.

\begin{theorem}\label{thm6}
Let $q=3$, $N=(q^m-1)/(q-1)$ and $m\geq 2$. Let $0\leq \ell \leq m-1$. Then the reverse code $\C_{(q, m, \ell)}^{-}=\C_{(q, m, m-1-\ell)}$. Furthermore, $\C_{(q, m, \ell)}$ and $\C_{(q, m, m-1-\ell)}$ have the same parameters.  	
\end{theorem}

When $q>3$, the conclusion of Theorem \ref{thm6} does not hold. We have the following example to illustrate.

\begin{example}
Let $q=5$ and $m=3$. Then $n=(q^m-1)/(q-1)=31$. It is easily verified that
$T_{(5, 3, 0)}=\{1, 5, 25\}$ and $T_{(4, 3, 2)}=\{49, 69, 73, 89, 93, 97, 109, 113, 117, 121 \}$. Let $\beta$ be a primitive element of $\gf(5^3)$ with $\beta^3+3\beta+3=0$. Then $\C_{(5, 3, 0)}$ is a $5$-ary $2$-constacyclic code with length $31$ and generator polynomial $x^3+3x+3$; $\C_{(5, 3, 2)}$ is a $5$-ary $2$-constacyclic code with length $31$ and generator polynomial $x^{10} + 3x^9 + 3x^8 + x^7 + 3x^6 + 2x^5 + 2 x^4 + 4 x^3 + x^2 + 2 x + 4$. By Magma, $\C_{(5, 3,0)}$ is a $5$-ary $[31, 28,3]$ distance-optimal code and $\C_{(5, 3, 2)}$ is a $5$-ary $[31, 21, 5]$ linear code.
\end{example}

Theorem \ref{thm6} shows that the ternary negacyclic codes $\C_{(3, m, \ell)}$ are very special, and we will apply them to construct ternary self-dual codes in the Section \ref{sec5}.

Below we study the parameters of $q$-ary constacyclic codes $\C_{(q, m, \ell)}$. By Equations (\ref{definingset1}) and (\ref{definingset2}), we deduce that the defining set $D_{(q, m, \ell)}$ in (\ref{definingset1}) of the constacyclic projective Reed-Muller codes is given by 
\begin{equation}\label{EEQQ0}
D_{(q, m, \ell)}=\bigcup_{i=0}^{\ell} T_{(q, m, i)}.	
\end{equation}

\begin{lemma}\label{LL::6}
Let $\C_i$ be a $q$-ary $\lambda$-constacyclic code with length $n$ and defining set $T_i$ with respect to the primitive $N$-th root of unity $\beta$, where $i=\{0, 1\}$. Then the following hold.
\begin{enumerate}
\item $\C_1\cap \C_2$ is a $q$-ary $\lambda$-constacyclic code with length $n$ and defining set $T_1 \cup T_2$ with respect to the primitive $N$-th root of unity $\beta$.
\item $\C_1+ \C_2$ is a $q$-ary $\lambda$-constacyclic code with length $n$ and defining set $T_1 \cap  T_2$ with respect to the primitive $N$-th root of unity $\beta$, where $\C_1+\C_2=\{\bc_1+\bc_2:~\bc_1\in \C_1, ~\bc_2\in \C_2\}$.
\item $\C_1 \subseteq \C_2$ if and only if $T_2\subseteq T_1$.
\end{enumerate}
\end{lemma}

\begin{proof}
The proof of this lemma is similar to the proof of a similar lemma for cyclic codes \cite{HP2003}, and we omit it.	
\end{proof}

By Lemma \ref{SDW}, Lemma \ref{LL::6} and Equation (\ref{EEQQ0}), we can directly obtain the following results.

\begin{theorem}
Let $1\leq \ell \leq m-2$. Then 	
$$\CPRM(q, m, m-2-\ell)=\CPRM(q, m, m-1-\ell) \cap \C_{(q, m, \ell)}$$
and $d(\C_{(q, m, \ell)})\leq 3\cdot q^{\ell}$.
\end{theorem}

Below, we study the parameters of this class of constacyclic codes $\C_{(q, m, \ell)}$. To determine the dimension of $\C_{(q, m, \ell)}$, we need the following lemma, which directly follows from \cite[Lemma 5]{Sorensen}.

\begin{lemma} \label{lem::5}
Let $0\leq \ell \leq m-1$. Then
$$
|T_{(q, m, \ell)}|=\sum_{h=0}^{m}(-1)^h \binom{m}{h}\binom{(q-1)\ell-h q+m}{1+(q-1)\ell-h q}.	
$$
\end{lemma}

The dimensions of $\C_{(q, m ,\ell)}$ and $\C_{(q, m ,\ell)}^\perp$ are documented in the following lemma.

\begin{lemma}\label{lem::6}
Let $0\leq \ell \leq m-1$. Then
$$
\dim( \C_{(q, m, \ell)})=n-\sum_{h=0}^{m}(-1)^h \binom{m}{h}\binom{(q-1)\ell-h q+m}{1+(q-1)\ell-h q}	
$$
and
$$
\dim( \C_{(q, m, \ell)}^\perp)=\sum_{h=0}^{m}(-1)^h \binom{m}{h}\binom{(q-1)\ell-h q+m}{1+(q-1)\ell-h q}.	
$$
\end{lemma}

\begin{proof}
  By definition, $\dim(\C_{(q, m, \ell)})=n-|T_{(q, m, \ell)}|$ and $\dim(\C_{(q, m, \ell)}^\perp)=|T_{(q, m, \ell)}|$. The desired results follow from Lemma \ref{lem::5}.
\end{proof}

To lower bound the minimum distance of the $\lambda$-constacyclic code $\C_{(q, m, \ell)}$, we need the following lemmas.

\begin{lemma}\label{lem::7}
Let $m\geq 3$ be odd. Then the following hold.
\begin{enumerate}
\item For any $0\leq \ell \leq (m-3)/2$,
$$
\{2q^\ell-1+(q^{(m+1)/2}-1)i:~0\leq i\leq 2q^\ell-1 \}	\subseteq T_{(q, m, \ell)}.
$$
\item If $\ell=(m-1)/2$,
$$
\{q^{m-1}+(q^{(m-1)/2}-1)i:~1\leq i\leq q^\ell \}	\subseteq T_{(q, m, \ell)}.
$$
\item For any $(m+1)/2 \leq \ell \leq m-1$,
$$
\{q^m-(q-1)q^{m-1-\ell}-(q^{(m+1)/2}-1)i:~0\leq i\leq (q-1)q^{m-1-\ell}-1 \}	\subseteq T_{(q, m, \ell)}.
$$
\end{enumerate}
\end{lemma}

\begin{proof}
1. For any $0\leq i\leq 2q^{\ell}-1$,
\begin{equation}\label{EQ::22}
2q^{\ell}-1+(q^{(m+1)/2}-1)i=i q^{(m+1)/2}+2q^{\ell}-1-i.	
\end{equation}
Suppose $i=i_1 q^\ell+i_0$, where $i_1\in \{0,1\}$ and $i_0\in \{0,1,\cdots, q^\ell-1\}$. By Equation (\ref{EQ::22}), we obtain that
\begin{equation}\label{EQ::23}
2q^{\ell}-1+(q^{(m+1)/2}-1)i=(i_1q^{\ell}+i_0) q^{(m+1)/2}+(1-i_1)q^{\ell}+q^\ell-1-i_0.	
\end{equation}
By Equation (\ref{EQ::23}), we get that
    \begin{align*}
		&\wt_q(2q^{\ell}-1+(q^{(m+1)/2}-1)i)\notag \\
		=~&i_1+\wt_q(i_0)+1-i_1+\wt_q( q^{\ell}-1-i_0) \notag \\
		=~&1+\wt_q(i_0)+(q-1)\ell-\wt_q(i_0) \notag \\
		=~& 1+(q-1)\ell.
	\end{align*}
The desired result follows.

2. For any $1\leq i\leq q^{(m-1)/2}$,
\begin{align}\label{EQ::25}
q^{m-1}+(q^{(m-1)/2}-1)i =q^{m-1}+(i-1)q^{(m-1)/2}+q^{(m-1)/2}-1-(i-1).
\end{align}
By Equation (\ref{EQ::25}), we have that
\begin{align*}
&\wt_q( q^{m-1}+(q^{(m-1)/2}-1)i  )\notag \\
=~&1+\wt_q(i-1)+(q-1)((m-1)/2)-\wt_q(i-1)\notag \\
=~&1+(q-1)[(m-1)/2].
\end{align*}
The desired result follows.

3. Since $(m+1)/2\leq \ell \leq m-1$, $0\leq m-1-\ell\leq (m-3)/2$. Suppose $i=i_1 q^{m-1-\ell}+i_0$, where $0\leq i_1\leq q-2$ and $0\leq i_0\leq q^{m-1-\ell}-1$. Then
\begin{align}\label{EQ::27}
&(q-1)q^{m-1-\ell}-1+(q^{(m+1)/2}-1)i\notag\\
=~&i q^{(m+1)/2}+(q-1)q^{m-1-\ell}-1-i\notag \\
=~&(i_1 q^{m-1-\ell}+i_0)q^{(m+1)/2}+(q-2-i_1) q^{m-1-\ell}+ q^{m-1-\ell}-1-i_0.
\end{align}
By Equation (\ref{EQ::27}), we get that
\begin{align}\label{EQ::28}
&~\wt_q((q-1)q^{m-1-\ell}-1+(q^{(m+1)/2}-1)i )\notag \\
=&~	i_1+\wt_q(i_0)+q-2-i_1+\wt_q(q^{m-1-\ell}-1-i_0)\notag \\
=&~q-2+(q-1)(m-1-\ell).
\end{align}
 By Equation (\ref{EQ::28}),
 \begin{align*}
 &~\wt_q( q^m-(q-1)q^{m-1-\ell}-(q^{(m+1)/2}-1)i) \notag\\	
 =&~(q-1)m-\wt_q((q-1)q^{m-1-\ell}-1+(q^{(m+1)/2}-1)i )\notag \\
=&~1+(q-1)\ell.
 \end{align*}
The desired result follows.	
\end{proof}

Similar to Lemma \ref{lem::7}, we can prove the following lemmas.

\begin{lemma}\label{lem::8}
Let $m\geq 4$ and $m\equiv 0\pmod{4}$. Then the following hold.
\begin{enumerate}
\item For any $0\leq \ell \leq (m-4)/2$,
$$
\{2q^\ell-1+(q^{(m+2)/2}-1)i:~0\leq i\leq 2q^\ell-1 \}	\subseteq T_{(q, m, \ell)}.
$$
\item If $\ell=(m-2)/2$,
$$
\{q^{m-1}+(q^{(m-2)/2}-1)i:~1\leq i\leq q^\ell \}	\subseteq T_{(q, m, \ell)}.
$$

\item If $\ell=m/2$,
$$
\{2q^{m-1}-1-(q^{(m-2)/2}-1)i:~1\leq i\leq q^{(m-2)/2} \}	\subseteq T_{(q, m, \ell)}.
$$

\item For any $(m+2)/2 \leq \ell \leq m-1$,
$$
\{q^m-(q-1)q^{m-1-\ell}-(q^{(m+2)/2}-1)i:~0\leq i\leq (q-1)q^{m-1-\ell}-1 \}	\subseteq T_{(q, m, \ell)}.
$$
\end{enumerate}
\end{lemma}

\begin{lemma}\label{lem::9}
Let $m\geq 6$ and $m\equiv 2\pmod{4}$. Then the following hold.
\begin{enumerate}
\item For any $0\leq \ell \leq (m-6)/2$,
$$
\{2q^\ell-1+(q^{(m+4)/2}-1)i:~0\leq i\leq 2q^\ell-1 \}	\subseteq T_{(q, m, \ell)}.
$$
\item If $\ell=(m-4)/2$,
$$
\{q^{m-1}+(q^{(m-4)/2}-1)i:~1\leq i\leq q^\ell \}	\subseteq T_{(q, m, \ell)}.
$$

\item If $\ell=(m-2)/2$,
$$
\{q^{m-1}+(q-1)q^{m-2}+(q^{(m-4)/2}-1)i:~1\leq i\leq q^{(m-4)/2} \}	\subseteq T_{(q, m, \ell)}.
$$

\item If $\ell=m/2$,
$$
\{q^{m-1}+(q-1)q^{m-2}+(q-1)q^{m-3}+(q^{(m-4)/2}-1)i:~1\leq i\leq q^{(m-4)/2} \}	\subseteq T_{(q, m, \ell)}.
$$

\item If $\ell=(m+2)/2$,
$$
\{2q^{m-1}-1-(q^{(m-4)/2}-1)i:~1\leq i\leq q^{(m-4)/2} \}	\subseteq T_{(q, m, \ell)}.
$$

\item For any $(m+4)/2 \leq \ell \leq m-1$,
$$
\{q^m-(q-1)q^{m-1-\ell}-(q^{(m+4)/2}-1)i:~0\leq i\leq (q-1)q^{m-1-\ell}-1 \}	\subseteq T_{(q, m, \ell)}.
$$
\end{enumerate}
\end{lemma}

To lower bound the minimum distance of the $\lambda^{-1}$-constacyclic code $\C_{(q, m, \ell)}^\perp$, we need the following lemmas.

\begin{lemma}\label{lem::10}
Let $m\geq 3$ and $0\leq \ell \leq m-1$. Then
$$
\{[1+(q-1)i] \bmod N:\  -(q^{m-1-\ell}-1)\leq i\leq 2(q^\ell-1)/(q-1)-1 \}	\subseteq \Omega_{(r, n)}^{(1)} \backslash T_{(q, m , \ell)}.
$$
\end{lemma}

\begin{proof}
If $0\leq i\leq 2(q^\ell-1)/(q-1)-1$, then $1\leq 1+(q-1)i\leq 2q^\ell-q$. It follows that
$$
\wt_q(1+(q-1)i)<(q-1)\ell+1
$$
for each $0\leq i\leq 2(q^\ell-1)/(q-1)-1$. It then follows that
$$
1+(q-1)i\in \Omega_{(r, n)}^{(1)} \backslash T_{(q, m , \ell)}.	
$$
If $-(q^{m-1-\ell}-1)\leq i \leq -1$, it is easily checked that
$$
\wt_q([1+(q-1)i] \bmod{N})=(q-1)m-\wt_q((q-1)(-i) -1).	
$$
Therefore, to prove the desired result, we only need to prove
$$
\wt_q((q-1)i -1 )\leq (q-1)(m-1-\ell)-1
$$
for each $1\leq i\leq q^{m-1-\ell}-1$. Suppose
$$
	(q-1)i-1=i_{m-1-\ell} q^{m-1-\ell}+\cdots+i_1q+i_0,
$$
where $i_j\in \{0,1,\cdots, q-1\}$. Notice that
$$
	(q-1)(q^{m-1-\ell}-1)-1=(q-2)q^{m-1-\ell}+(q-1)q^{m-2-\ell}+\cdots+(q-1)q
$$
and $\wt_q((q-1)(q^{m-1-\ell}-1)-1 )=(q-1)(m-1-\ell)-1$. If $i_{m-1-\ell}=q-2$ and $i<q^{m-1-\ell}-1$, then there is $v\in \{1, 2,\cdots, m-2-\ell\}$ such that $0\leq i_v\leq q-2$. Consequently,
\begin{align}\label{EQ::48}
\wt_q((q-1)i-1)	&\leq 2(q-2)+(q-1)(m-2-\ell) \notag \\
&= q-3+(q-1)(m-1-\ell).
\end{align}
If $0\leq i_{m-1-\ell}\leq q-3$, we get that
\begin{equation}\label{EQ::49}
\wt_q((q-1)i-1)\leq (q-3)+(q-1)(m-1-\ell).	
\end{equation}
Notice that $\wt_q((q-1)i-1)\equiv -1\pmod{q-1}$, by Equations (\ref{EQ::48}) and (\ref{EQ::49}), we obtain that
$$
\wt_q((q-1)i-1)\leq (q-1)(m-1-\ell)-1.
$$
The desired result follows. This completes the proof.
\end{proof}

Collecting the lemmas above, we arrive at the following theorem.

\begin{theorem}\label{Thm::11}
Let $q>2$ be a prime power. Let $m\geq 2$ and $0\leq \ell \leq m-1$. Then the $\lambda$-constacyclic code $\C_{(q, m, \ell)}$ has parameters $[(q^m-1)/(q-1), k, d\geq d_{(q, m,\ell)}]$, where
$$
k=(q^m-1)/(q-1)-\sum_{h=0}^{m}(-1)^h \binom{m}{h}\binom{(q-1)\ell-h q+m}{1+(q-1)\ell-h q}
$$
and
\begin{align*}
d_{(q, m,\ell)}=\begin{cases}
2q^\ell+1~&{\rm if~}m~{\rm is~odd~ and~}0\leq \ell \leq (m-3)/2,\\
~&{\rm or~} m \equiv 0 ~({\rm mod}~4)~{\rm and}~0\leq \ell\leq (m-4)/2,\\
 &{\rm or~}m \equiv 2 ~({\rm mod}~4)~{\rm and}~0\leq \ell\leq (m-6)/2,\\
q^{(m-4)/2}+1~&{\rm if~}m \equiv 2 ~({\rm mod}~4)~and~(m-4)/2\leq \ell \leq (m+2)/2,\\
q^{(m-2)/2}+1~&{\rm if~}m \equiv 0 ~({\rm mod}~4)~{\rm and}~\ell \in \{ (m-2)/2,m/2\},\\
 q^{(m-1)/2}+1~&{\rm if}~m~{\rm is~odd~and~}\ell=(m-1)/2,\\
 (q-1)q^{m-1-\ell}+1~&{\rm otherwise}.
\end{cases}
\end{align*}
The $\lambda^{-1}$-constacyclic code $\C_{(q, m, \ell)}^\perp$ has parameters $[(q^m-1)/(q-1), k^\perp, d^\perp \geq d^\perp_{(q, m,\ell)}]$, where
$$
k^\perp=\sum_{h=0}^{m}(-1)^h \binom{m}{h}\binom{(q-1)\ell-h q+m}{1+(q-1)\ell-h q}
$$
and $d^\perp_{(q, m,\ell)}=q^{m-1-\ell}+2(q^\ell-1)/(q-1)$.
\end{theorem}

\begin{proof}
The dimensions of two codes $\C_{(q, m,\ell)}$ and $\C_{(q, m,\ell)}^\perp$ follow from Lemma \ref{lem::6}. The lower bounds on the minimum distances of $\C_{(q, m, \ell)}$ follow from the constacyclic BCH bound and Lemmas \ref{lem::7}, \ref{lem::8} and \ref{lem::9}. The lower bound on the minimum distance of $\C_{(q, m, \ell)}^\perp$ follows from the constacyclic BCH bound and Lemma \ref{lem::10}.	 This completes the proof.
\end{proof}

The experimental data in Table \ref{Tab-code1} shows that the constacyclic codes $\C_{(q, m, \ell)}$ and $\C_{(q, m, \ell)}^\perp$ have very good parameters.

\begin{table*}
\begin{center}
\caption{Examples of the codes $\C_{(q, m, \ell)}$ and $\C_{(q, m, \ell)}^\perp$}\label{Tab-code1}
\begin{tabular}{|c|c|c|c|c|c|c|} \hline
$q$ & $m$ & $\ell$ & $\C_{(q, m, \ell)}$ & Optimality & $\C_{(q, m, \ell)}^\perp$ & Optimality  \\ \hline
$3$ & $2$ & $0,1$ & $[4,2,3]$ & optimal linear code & $[4,2,3]$ & optimal linear code  \\ \hline
$3$ & $3$ & $0,2$ & $[13,10,3]$ & optimal linear code & $[13,3,9]$ & optimal linear code  \\ \hline
$3$ & $3$ & $1$ & $[13,6,6]$ & optimal linear code & $[13,7,5]$ & optimal linear code  \\ \hline
$3$ & $4$ & $0,3$ & $[40,36,3]$ & optimal linear code & $[40,4,27]$ & optimal linear code  \\ \hline
$3$ & $4$ & $1,2$ & $[40,24,9]$ & best linear code known & $[40,16,15]$ & best linear code known  \\ \hline
$3$ & $5$ & $0,4$ & $[121,116,3]$ & optimal linear code & $[121,5,81]$ & optimal linear code  \\ \hline
$3$ & $5$ & $1$ & $[121,91,9]$ & $d_{\rm best}=11$ & $[121,30,45]$ & $d_{\rm best}=46$  \\ \hline
$4$ & $2$ & $0$ & $[5,3,3]$ & optimal linear code & $[5,2,4]$ & optimal linear code  \\ \hline
$4$ & $2$ & $1$ & $[5,2,4]$ & optimal linear code & $[5,3,3]$ & optimal linear code  \\ \hline
$4$ & $3$ & $0$ & $[21,18,3]$ & optimal linear code & $[21,3,16]$ & optimal linear code  \\ \hline
$4$ & $3$ & $1$ & $[21,9,9]$ & optimal linear code & $[21,12,7]$ & optimal linear code  \\ \hline
$4$ & $3$ & $2$ & $[21,15,4]$ & $d_{\rm best}=5$ & $[21,6,12]$ & optimal linear code  \\ \hline
$5$ & $2$ & $0$ & $[6,4,3]$ & optimal linear code & $[6,2,5]$ & optimal linear code  \\ \hline
$5$ & $2$ & $1$ & $[6,2,5]$ & optimal linear code & $[6,4,3]$ & optimal linear code  \\ \hline
$5$ & $3$ & $0$ & $[31,28,3]$ & optimal linear code & $[31,3,25]$ & optimal linear code  \\ \hline
$5$ & $3$ & $1$ & $[31,13,13]$ & best linear code known & $[31,18,9]$ & best linear code known  \\ \hline
$5$ & $3$ & $2$ & $[31,21,5]$ & $d_{\rm best}=7$ & $[31,10,15]$ & best linear code known  \\ \hline
$7$ & $2$ & $0$ & $[8,6,3]$ & optimal linear code & $[8,2,7]$ & optimal linear code  \\ \hline
$7$ & $2$ & $1$ & $[8,2,7]$ & optimal linear code & $[8,6,3]$ & optimal linear code  \\ \hline
$7$ & $3$ & $0$ & $[57,54,3]$ & optimal linear code & $[57,3,49]$ & optimal linear code  \\ \hline
$7$ & $3$ & $1$ & $[57,24,21]$ & best linear code known & $[57,33,13]$ & $d_{\rm best}=15$   \\ \hline
$7$ & $3$ & $2$ & $[57,36,7]$ & $d_{\rm best}=13$ & $[57,21,21]$ & $d_{\rm best}=24$   \\ \hline
\end{tabular}
\end{center}
\end{table*}

\section{The subcodes of the constacyclic codes $\C_{(q, m, \ell)}$} \label{sec5}

In this section, we study the parameters of several classes of subcodes of the constacyclic codes $\C_{(q, m, \ell)}$. For any positive integer $i$, let $v_q(i)$ be a nonnegative integer $j$ such that $q^j \mid i$ and $q^{j+1}\nmid i$. We have the following lemmas.

\begin{lemma}\label{lem26}
Let $m\geq 5$ and $1\leq i\leq q^m-2$. Then $\wt_q(i-1)=\wt_q(i)-1+(q-1)v_q(i)$.	
\end{lemma}

\begin{proof}
Suppose $v_q(i)=s$. Then $i=[i_0+i_1q+\cdots+i_t q^t] q^s$, where $1\leq i_0\leq q-1$, $0\leq t\leq m-1-s$, $0\leq i_j\leq q-1$ for $1\leq j\leq t$. It follows that
$$i-1=[i_0-1+i_1q+\cdots+i_t q^t] q^s+q^s-1.$$
Consequently,
$$\wt_q(i-1)=(i_0-1)+i_1+\cdots+i_t+(q-1)s=\wt_q(i)-1+(q-1)s.$$
This completes the proof.
\end{proof}

\begin{lemma}\label{lem::12}
Let $q\geq 3$ and $m\geq 5$ be odd. Then
\begin{align*}
&\wt_q( q^{m-1}+(q^{(m-1)/2}-1)i )\\
=&\begin{cases}
1+(q-1)[(m-1)/2-v_q(-i)]~&{\rm if}~-q^{(m-1)/2}+1\leq i\leq -1,\\
1~&{\rm if}~i=0,\\	
1+(q-1)[(m-1)/2]~&{\rm if}~1\leq i\leq q^{(m-1)/2},\\
1+(q-1)(m-1)~&{\rm if}~i= q^{(m-1)/2}+1,\\
1+(q-1)[(m-1)/2+v_q(i-q^{(m-1)/2}-1 )]~&{\rm if}~q^{(m-1)/2}+2\leq i\leq  2q^{(m-1)/2}.
\end{cases}
\end{align*}
\end{lemma}

\begin{proof}
By Lemma \ref{lem::7}, $$\wt_q(q^{m-1}+(q^{(m-1)/2}-1)i)=1+(q-1)[(m-1)/2]$$ for each $1\leq i\leq q^{(m-1)/2}$. It is easily checked that $\wt_q(q^{m-1})=1$ and $$\wt_q(q^{m-1}+(q^{(m-1)/2}+1)(q^{(m-1)/2}-1) )=1+(q-1)(m-1).$$
Notice that
$$q^{m-1}-(q^{(m-1)/2}-1)i=q^{(m-1)/2}[q^{(m-1)/2}-1-(i-1)]+i $$
and
$$q^{m-1}+(q^{(m-1)/2}-1)(q^{(m-1)/2}+1+i)=2q^{m-1}+(i-1)q^{(m-1)/2}+q^{(m-1)/2}-1-i.$$
Then, if $1\leq i\leq q^{(m-1)/2}-1$,
$$\wt_q(q^{m-1}-(q^{(m-1)/2}-1)i)=(q-1)[(m-1)/2]-\wt_q(i-1)+\wt_q(i)$$
and
$$\wt_q(q^{m-1}+(q^{(m-1)/2}-1)(q^{(m-1)/2}+1+i))=(q-1)[(m-1)/2]-\wt_q(i)+\wt_q(i-1)+2.$$
By Lemma \ref{lem26},
$$\wt_q(q^{m-1}-(q^{(m-1)/2}-1)i)=1+(q-1)[(m-1)/2-v_q(i)]$$
and
$$\wt_q(q^{m-1}+(q^{(m-1)/2}-1)(q^{(m-1)/2}+1+i))=1+(q-1)[(m-1)/2+v_q(i)].$$
This completes the proof.
\end{proof}

The parameters of the first class of subcodes of $\C_{(q, m, (m-1)/2)}$ are documented in the following theorem.

\begin{theorem}
	Let $m\geq 5$ be odd and ${\mathcal S}_1(\C)$ be the $q$-ary $\lambda$-constacyclic code of  length $n$ with defining set $T_{(q, m, (m-1)/2)} \cup C_1^{(q, N)}$ with respect to the primitive $N$-th root of unity $\beta$. Then ${\mathcal S}_1(\C)$ has parameters $\left [(q^m-1)/(q-1), \dim(\C_{(q, m, (m-1)/2)})-m, d\geq q^{(m-1)/2}+q+1 \right ]$, and the dual code $\mathcal{S}_1(\C)^\perp$ has parameters $$\left [(q^m-1)/(q-1), \dim(\C_{(q, m, (m-1)/2)}^\perp)+m, d^\perp \geq q^{(m-1)/2} \right ].$$
\end{theorem}

\begin{proof}
Notice that $\wt_q(i)=1\neq 1+(q-1)[(m-1)/2]$ for $i\in C_1^{(q, N)}$. Then $$C_1^{(q, N)} \cap T_{(q, m, (m-1)/2)}=\emptyset.$$
Clearly, $|C_1^{(q, N)}|=m$. Then $$\dim(\mathcal{S}_1(\C))=N-|T_{(q, m, (m-1)/2)}|-m= \dim(\C_{(q, m, (m-1)/2)})-m$$
and
$$\dim(\mathcal{S}_1(\C)^\perp)=|T_{(q, m, (m-1)/2)}|+m= \dim(\C_{(q, m, (m-1)/2)}^\perp)+m.$$
	By Lemma \ref{lem::12},
	$$\left \{ q^{m-1}+(q^{(m-1)/2}-1)i: -(q-1)\leq i\leq q^{(m-1)/2} \right \} \subseteq T_{(q, m, (m-1)/2)} \cup C_1^{(q, N)}.$$
	By the constacyclic BCH bound, $d(\mathcal{S}_1(\C))\geq q^{(m-1)/2}+q+1$. By Lemma \ref{lem::10},
	$$\{q^m-(q-1)i: ~1\leq i\leq q^{(m-1)/2}-1  \} \subseteq \Omega_{(r, n)}^{(1)} \backslash T_{(q, m, (m-1)/2)}.  $$
	It is easily verified that $q^m-(q-1)i \notin C_1^{(q, N)}$ for each $1\leq i\leq q^{(m-1)/2}-1$. Therefore,
		$$\{q^m-(q-1)i: ~1\leq i\leq q^{(m-1)/2}-1  \} \subseteq \Omega_{(r, n)}^{(1)} \backslash (T_{(q, m, (m-1)/2)} \cup C_1^{(q, N)}). $$
	By Lemma \ref{lem::3}, $d(\mathcal{S}_1(\C)^\perp )\geq q^{(m-1)/2}$. This completes the proof.
\end{proof}

The parameters of the second class of subcodes of $\C_{(q, m, (m-1)/2)}$ are documented in the following theorem.

\begin{theorem}
	Let $m\geq 5$ be odd. Let ${\mathcal S}_2(\C)$ be the $q$-ary $\lambda$-constacyclic code of length $n$ with defining set $T_{(q, m, (m-1)/2)} \cup C_1^{(q, N)} \cup C_{2q^{m-1}-1}^{(q, N)}$ with respect to the primitive $N$-th root of unity $\beta$. Then ${\mathcal S}_2(\C)$ has parameters $$\left [(q^m-1)/(q-1), \dim(\C_{(q, m, (m-1)/2)})-2m, d\geq q^{(m-1)/2}+2q+1 \right ],$$
	and $\mathcal{S}_2(\C)^\perp$ has parameters
	$$\left [(q^m-1)/(q-1), \dim(\C_{(q, m, (m-1)/2)}^\perp)+2m, d^\perp \geq (q-1)q^{(m-3)/2}+1 \right ].$$
\end{theorem}

\begin{proof}
Notice that $\wt_q(i)=1+(q-1)(m-1)$ for $i\in C_{2q^{m-1}-1}^{(q, N)}$. Then $$C_{2q^{m-1}-1}^{(q, N)} \cap (T_{(q, m, (m-1)/2)}\cup C_1^{(q, N)})=\emptyset.$$
It is easily verified that $|C_{2q^{m-1}-1}^{(q, N)}|=|C_{q-2}^{(q, N)}|=m$. Then $$\dim(\mathcal{S}_2(\C))=N-|T_{(q, m, (m-1)/2)}|-2m= \dim(\C_{(q, m, (m-1)/2)})-2m$$
and
$$\dim(\mathcal{S}_2(\C)^\perp)=|T_{(q, m, (m-1)/2)}|+2m= \dim(\C_{(q, m, (m-1)/2)}^\perp)+2m.$$
	By Lemma \ref{lem::12},
	$$\left \{ q^{m-1}+(q^{(m-1)/2}-1)i: -(q-1)\leq i\leq q^{(m-1)/2}+q \right \} \subseteq T_{(q, m, (m-1)/2)} \cup C_1^{(q, N)}\cup C_{2q^{m-1}-1}^{(q, N)}.$$
	By the constacyclic BCH bound, $d(\mathcal{S}_2(\C))\geq q^{(m-1)/2}+2q+1$. By Lemma \ref{lem::7},
	$$\{q^m-(q-1)i: ~1\leq i\leq q^{(m-1)/2}-1  \} \subseteq \Omega_{(r, n)}^{(1)} \backslash T_{(q, m, (m-1)/2)}.  $$
	It is easily verified that $q^m-(q-1)i \notin C_1^{(q, N)}$ for each $1\leq i\leq q^{(m-1)/2}-1$. Therefore,
		$$\{q^m-(q-1)q^{(m-3)/2}-(q^{(m+1)/2}-1)i: 0\leq i\leq (q-1)q^{(m-3)/2}-1  \} \subseteq  T_{(q, m, (m+1)/2)}. $$
	By Lemma \ref{lem::3}, $d(\mathcal{S}_2(\C)^\perp )\geq (q-1)q^{(m-3)/2}+1$. This completes the proof.
\end{proof}

The parameters of the third class of subcodes of $\C_{(q, m, \ell )}$ are documented in the following theorem.

\begin{theorem}
	Let $m\geq 4$ and $0\leq \ell \leq \lfloor(m-2)/2\rfloor$. Let ${\mathcal S}_{3,\ell}(\C)$ be the $q$-ary $\lambda$-constacyclic code of length $n$ with defining set $T_{(q, m, \ell)} \cup T_{(q, m, m-1-\ell)}$ with respect to the primitive $N$-th root of unity $\beta$. Then ${\mathcal S}_{3,\ell}(\C)$ has parameters
	$$\left [(q^m-1)/(q-1), (q^m-1)/(q-1)-k, \geq d \right ],$$
	and the dual code $\mathcal{S}_{3, \ell}(\C)^\perp$ has parameters $\left [(q^m-1)/(q-1), k, \geq d^\perp  \right ]$, where
	$$k=\sum_{h=0}^{m}(-1)^h \binom{m}{h}\left[ \binom{(q-1)\ell-h q+m}{1+(q-1)\ell-h q} +\binom{(q-1)(m-1-\ell)-h q+m}{1+(q-1)(m-1-\ell)-h q}\right],$$
	$$d=\begin{cases}
(q-1)q^\ell+1~&{\rm if}~0\leq \ell \leq (m-5)/2,	\\
(q-1)q^\ell+1~&{\rm if}~m~{\rm is ~odd~ and~}\ell = (m-3)/2,	\\
(q-1)q^\ell+1~&{\rm if}~m\equiv 0~({\rm mod}~4)~{\rm and}~\ell = (m-4)/2,	\\
q^\ell+1~&{\rm if}~m\equiv 2~({\rm mod}~4)~{\rm and}~\ell = (m-4)/2,\\
q^\ell+1~&{\rm if}~m\equiv 0~({\rm mod}~4)~{\rm and}~\ell = (m-2)/2,\\	
q^{\ell-1}+1~&{\rm if}~m\equiv 2~({\rm mod}~4)~{\rm and}~\ell = (m-2)/2,\\
\end{cases}$$
and
$$d^\perp=\begin{cases}
q^{(m-1)/2}+1~&{\rm if}~m~{\rm is~odd},	\\
q^{(m-2)/2}+1~&{\rm if}~m\equiv 0~({\rm mod}~4)~{\rm and}~0\leq \ell \leq (m-4)/2,	\\
(q-1)q^{(m-4)/2}+1~&{\rm otherwise}.
\end{cases}$$
\end{theorem}

\begin{proof}
It is clear that $\dim({\mathcal S}_{3,\ell}(\C))=n-|T_{(q, m, \ell)}|-|T_{(q, m, m-1-\ell)}|$ and
$$\dim(\mathcal{S}_{3,\ell}(\C))=|T_{(q, m, \ell)}|+|T_{(q, m, m-1-\ell)}|.$$
The desired dimensions follow from Lemma \ref{lem::5}. By Lemma \ref{LL::6}, $$\mathcal{S}_{3, \ell}(\C)=\C_{(q, m, \ell)} \cap \C_{(q, m, m-1-\ell)}.$$
It follows that $d(\mathcal{S}_{3, \ell}(\C))\geq \max\{d(\C_{(q, m, \ell)}), d(\C_{(q, m, m-1-\ell)})  \}$. The desired lower bounds on minimum distances of $\mathcal{S}_{3, \ell}(\C)$ follow from Theorem \ref{Thm::11}. We now lower bound the minimum distance of $\mathcal{S}_{3, \ell}(\C)^\perp$.

{\it Case 1}: $m$ is odd. Notice that $T_{(q, m, (m-1)/2)}\subseteq \Omega_{(r, n)}^{(1)}\backslash (T_{(q, m, \ell)} \cup T_{(q, m, m-1-\ell)})$. By Lemma \ref{lem::7}, we get that
$\{q^{m-1}+(q^{(m-1)/2}-1)i: ~1\leq i\leq q^{(m-1)/2} \}\subseteq \Omega_{(r, n)}^{(1)}\backslash (T_{(q, m, \ell)} \cup T_{(q, m, m-1-\ell)}) $. By Lemma \ref{lem::3}, we obtain that $d(\mathcal{S}_{3,\ell}(\C)^\perp)\geq q^{(m-1)/2}+1$.

{\it  Case 2}: $m\equiv 0~({\rm mod}~4)$. If $0\leq \ell \leq (m-4)/2$, by Lemma \ref{lem::8},
$$\{q^{m-1}+(q^{(m-2)/2}-1)i: ~1\leq i\leq q^{(m-2)/2} \}\subseteq \Omega_{(r, n)}^{(1)}\backslash (T_{(q, m, \ell)} \cup T_{(q, m, m-1-\ell)}) .$$ By Lemma \ref{lem::3}, we get that $d(\mathcal{S}_{3,\ell}(\C)^\perp)\geq q^{(m-2)/2}+1$. If $\ell=(m-2)/2$, $$T_{(q, m, (m+2)/2)}\subseteq \Omega_{(r, n)}^{(1)}\backslash (T_{(q, m, \ell)} \cup T_{(q, m, m-1-\ell)}).$$
By Lemmas \ref{lem::3} and \ref{lem::8}, we get that $d(\mathcal{S}_{3,\ell}(\C)^\perp)\geq (q-1)q^{(m-4)/2}+1$.

{\it Case 3}: $m\equiv 2~({\rm mod}~4)$. It is easily verified that
 $$T_{(q, m, (m-2)/2)}\subseteq \Omega_{(r, n)}^{(1)}\backslash (T_{(q, m, \ell)} \cup T_{(q, m, m-1-\ell)})$$
 or  $T_{(q, m, (m-4)/2)}\subseteq \Omega_{(r, n)}^{(1)}\backslash (T_{(q, m, \ell)} \cup T_{(q, m, m-1-\ell)})$. By Lemmas \ref{lem::3} and \ref{lem::9}, we deduce that $$d(\mathcal{S}_{3,\ell}(\C)^\perp)\geq (q-1)q^{(m-4)/2}+1.$$
This completes the proof.
\end{proof}

\begin{example}
	Let $q=3$ and $m=4$. Then $n=(q^m-1)/(q-1)=40$. Let $\beta$ be a primitive element of $\gf(q^4)$ with $\beta^4+2\beta^3+2=0$. Then $\mathcal{S}_{3,0}(\C)$ is a ternary negacyclic code of length $n$ and generator polynomial $g(x)=x^8 + 2x^7 + x^5 + x^3 + 2x + 1$. By Magma, $\mathcal{S}_{3,0}(\C)$ is a ternary $[40,32,5]$ linear code, which is a distance-optimal code \cite{Grassl}. The dual code $\mathcal{S}_{3,0}(\C)^\perp$ is a ternary $[40,8,21]$ linear code, which has the best parameters known \cite{Grassl}.
\end{example}

\begin{example}
	Let $q=4$ and $m=4$. Then $n=(q^m-1)/(q-1)=85$. Let $\beta$ be a primitive element of $\gf(q^4)$ with $\beta^4 + \beta^3 + \omega \beta^2 + \omega \beta +\omega=0$, where $\omega$ is a primitive element of $\gf(4)$ with $\omega^2+\omega+1=0$.
	Then $\mathcal{S}_{3,0}(\C)$ is a quaternary $\omega$-constacyclic code of length $n$ and generator polynomial $g(x)=x^{14} + x^{13} + \omega x^{12} + x^{10} + \omega^2 x^9 + x^8 + \omega^2 x^7 + x^6 + x^5 + x^3 + \omega x^2 + x + \omega^2$. By Magma, $\mathcal{S}_{3,0}(\C)$ is a quaternary $[85,71,7]$ linear code, which has the best parameters known \cite{Grassl}. The dual code $\mathcal{S}_{3,0}(\C)^\perp$ is a quaternary $[85,14,46]$ linear code. The best quaternary linear code of length $85$ and dimension $14$ has minimum distance $48$ \cite{Grassl}.
\end{example}

When $q=3$ and $m$ is even, the fourth class of subcodes of $\C_{(q, m, \ell)}$ are studied in the following theorem.

\begin{theorem}
Let $m\geq 4$ be even. Let $\Xi_i=\{i, m-1-i\}$ and $j_i \in \Xi_i$ for each $0\leq i \leq (m-2)/2$. Let $\mathcal{S}_{4}(\C)$ be the ternary negacyclic code of length $n$ and defining set $$Z:=T_{(q, m, j_0)}\cup T_{(q, m, j_1)}\cup \cdots \cup T_{(q, m, j_{(m-2)/2})}.$$
Then $\mathcal{S}_4(\C)$ is a ternary $[(3^m-1)/2,(3^m-1)/4, \geq d ]$ self-dual code, where
$$d=\begin{cases}
3^{(m-2)/2}+3~&{\rm if}~m\equiv 0~({\rm mod}~4),\\
3^{(m-4)/2}+3~&{\rm if}~m\equiv 2~({\rm mod}~4),\\		
\end{cases}
$$
\end{theorem}

\begin{proof}
Recall $N=3^m-1$. Notice that $\wt_3(N-i)=1+2(m-1-\ell)$ for any $i\in T_{(3, m, \ell)}$. For $S\subset \Z_{N}$, define $-S=\{N-i: ~i\in S\}$. It is easily verified that $-T_{(3, m, \ell)}=T_{(3, m, m-1-\ell)}$. Consequently, $$-Z=T_{(q, m, m-1-j_0)}\cup T_{(q, m, m-1-j_1)}\cup \cdots \cup  T_{(q, m, m-1-j_{(m-2)/2})}.$$
By definition, $-Z \cup Z=\emptyset$ and $Z \cup -Z =\{1+2i:~0\leq i\leq n-1\}$. It follows that $\mathcal{S}_4(\C)$ is a self-dual code. Notice that $T_{(3, m, (m-2)/2)} \subseteq Z$ or $T_{(3, m, m/2)} \subseteq Z$. By Lemmas \ref{lem::8}, \ref{lem::9} and the constacyclic BCH bound,
$d(\mathcal{S}_4(\C))\geq d-1$. According to \cite[Theorem 1.4.5]{HP2003}, $d(\mathcal{S}_4(\C))\equiv 0~({\rm mod}~3)$. The desired lower bounds follow. This completes the proof.
\end{proof}

When $j_i=i$ for any $0\leq i\leq (m-2)/2$, the ternary negacyclic code $$\mathcal{S}_4(\C)=\CPRM(3, m, (m-2)/2).$$ By Lemma \ref{SDW}, $\mathcal{S}_4(\C)$ is a ternary $[(3^m-1)/2, (3^m-1)/4, 3^{m/2}]$ self-dual code. The following example illustrates that  $\mathcal{S}_4(\C)$ and $\C_{(1,n)}$, as well as the ternary self-dual codes constructed in \cite{XCDS} are different.

\begin{example}
Let $q=3$ and $m=4$. Then $n=(q^m-1)/(q-1)=40$. Let $\beta$ be a primitive element of $\gf(q^4)$ with $\beta^4 + 2 \beta^3 + 2=0$. Let $j_0=0$ and $j_1=2$. Then $\mathcal{S}_4(\C)$ has the defining set
$$\{ 1, 3, 9, 17, 23, 25, 27, 35, 41, 43, 47, 49, 51, 59, 61, 65, 67, 69, 73, 75 \}$$
with respect to the $\beta$. The defining set of ternary self-dual codes in \cite{XCDS} has a long continuous sequence of arithmetic progression. By Magma, $\mathcal{S}_4(\C)$ is a ternary $[40,20,9]$ self-dual code. Although $\mathcal{S}_4(\C)$ and $\C_{(1,n)}$ have the same parameters, the order of the automorphism group of $\mathcal{S}_4(\C)$ is $11520$, while the order of the automorphism group of $\C_{(1,n)}$ is $24261120$. Hence,  $\mathcal{S}_4(\C)$ and $\C_{(1,n)}$ are different. 	
\end{example}

\section{Concluding remarks}\label{sec6}

Two infinite classes of constacyclic codes over finite fields were constructed in this paper. These classes of constacyclic codes are interesting as they contain distance-optimal codes and linear codes with best known parameters (see Tables \ref{DS:1} and \ref{Tab-code1}). A summary of the main contributions of this paper goes as follows.

\begin{enumerate}
\item An infinite class of $q$-ary negacyclic codes $\C$ of length $(q^m-1)/2$ such that $d(\C)$ and $d(\C^\perp)$ both have a square-root-like lower bound was constructed. 	
\item An infinite class of $q$-ary constacyclic codes $\C$ of length $(q^m-1)/(q-1)$ such that $d(\C)$ and $d(\C^\perp)$ both have a square-root-like lower bound was constructed.

\item Two infinite classes of ternary negacyclic self-dual codes of length $(3^m-1)/2$ with a square-root-like lower bound were constructed.
\end{enumerate}
The lower bounds on the minimum distances of two classes of constacyclic codes and their duals developed in this paper are good. The reader is cordially  invited to improve the lower bounds or settle the minimum distances of these two classes of constacyclic codes.

\end{document}